\documentclass[pra,twocolumn,showpacs]{revtex4-2}

\usepackage{amsmath}
\usepackage{latexsym}
\usepackage{amssymb}
\usepackage{graphicx}
\usepackage[colorlinks=true, citecolor=blue, urlcolor=blue]{hyperref}
\usepackage{float}
\usepackage{amsfonts}
\usepackage{textcomp}
\usepackage{mathpazo}
\usepackage{comment}
\usepackage{xr}

\usepackage{bbm}

\usepackage{xcolor}
\definecolor{myurlcolor}{rgb}{0,0,0.4}
\definecolor{mycitecolor}{rgb}{0,0.5,0}
\definecolor{myrefcolor}{rgb}{0.5,0,0}
\usepackage{hyperref}
\hypersetup{colorlinks,
linkcolor=myrefcolor,
citecolor=mycitecolor,
urlcolor=myurlcolor}

\usepackage[draft]{fixme}
\usepackage{amsmath,bbm}
\usepackage{graphicx}
\usepackage{amsfonts}
\usepackage{amssymb}
\usepackage{amsmath, amssymb, amsthm,verbatim,graphicx,bbm}
\usepackage{mathrsfs}
\usepackage{color,xcolor,longtable}

\usepackage{xr}
\makeatletter
\newcommand*{\addFileDependency}[1]{
  \typeout{(#1)}
  \@addtofilelist{#1}
  \IfFileExists{#1}{}{\typeout{No file #1.}}
}
\makeatother

\newcommand*{\myexternaldocument}[1]{
    \externaldocument{#1}
    \addFileDependency{#1.tex}
    \addFileDependency{#1.aux}
}

\myexternaldocument{manuscript_PRL}

\newcommand{\beq}[0]{\begin{equation}}
\newcommand{\eeq}[0]{\end{equation}}

\newcommand{\one}{\leavevmode\hbox{\small1\normalsize\kern-.33em1}}

\def\be{\begin{equation}}
\def\ee{\end{equation}}
\def\ben{\begin{eqnarray}}
\def\een{\end{eqnarray}}
\def\eea{\end{array}}
\def\bea{\begin{array}}

\newcommand{\Tr}[1]{\mathrm{Tr}#1}
\newcommand{\bei}{\begin{itemize}}
\newcommand{\eei}{\end{itemize}}
\newcommand{\ket}[1]{|#1\rangle}
\newcommand{\bra}[1]{\langle#1|}

\newcommand{\proj}[1]{\ket{#1}\!\!\bra{#1}}

\newcommand{\I}{\mathbbm{1}}

\renewcommand{\emph}[1]{\textbf{#1}}

\makeatletter
\newtheorem*{rep@theorem}{\rep@title}
\newcommand{\newreptheorem}[2]{%
\newenvironment{rep#1}[1]{%
 \def\rep@title{#2 \ref{##1}}%
 \begin{rep@theorem}}%
 {\end{rep@theorem}}}
\makeatother

\theoremstyle{plain}
\newtheorem*{thm*}{Theorem}
\newreptheorem{thm}{Theorem}

\newtheorem{fakt}{Fact}

\theoremstyle{definition}

\theoremstyle{remark}

\usepackage[T1]{fontenc}

\begin{document}

\title{Distrustful quantum steering}
\author{Shubhayan Sarkar}
\email{shubhayan.sarkar@ulb.be}
\affiliation{Laboratoire d’Information Quantique, Université libre de Bruxelles (ULB), Av. F. D. Roosevelt 50, 1050 Bruxelles, Belgium}

\begin{abstract}	
Quantum steering is an asymmetric form of quantum nonlocality where one can trust the measurements of one of the parties. In this work, inspired by practical considerations we investigate the scenario if one can not fully trust their measurement devices but only up to some precision. We first find the effect of such an imprecision on standard device-dependent quantum tomography. We then utilise this result to compute the variation in the local bound of any general steering inequality depending on the amount of trust one puts in one of the party's measurement devices. This is particularly important as we show that even a small distrust on Alice might cause the parties to observe steerability even if the quantum state is unsteerable. Furthermore, this effect becomes more relevant when observing higher dimensional quantum steering. 
\end{abstract}


\maketitle

\section{Introduction} The existence of nonlocal correlations is one of the most fascinating aspects of quantum theory. An asymmetric form of quantum nonlocality is termed quantum steering where one observes the change of a local quantum state depending on the operations carried out by another space-like separated observer and the classical information shared
between the parties. The idea of quantum steering was conceived by Schr$\mathrm{\ddot{o}}$dinger in 1935 but was mathematically formalised after almost seventy years \cite{WJD07, CJWR09}.
Contrary to standard Bell scenario \cite{bel64}, in a steering scenario Alice on one side has complete knowledge about her measurement device, thus known as the trusted side, however for Bob, it is still a black box [see Refs. \cite{CS16, SteeringReview} for a review on quantum steering].

The trust in one of the parties makes quantum steering more robust to noise and detector inefficiencies when compared to Bell nonlocality  \cite{ GCH+19,DSU+21, PCS+15, LTB+14}. Consequently, it is much more simpler to observe quantum steering than Bell nonlocality in higher dimensional systems. Several steering witnesses have been proposed for higher dimensional systems \cite{Horodecki, SC18, Caval1,sarkar6, sarkar7}. Higher dimensional quantum steering has also been observed in experiments \cite{steexp1,steexp2,steeexp3}.

To observe quantum steerability in experiments, one needs to completely trust at least one of the measurement devices. However, one can not be completely confident that the measurements implemented are exactly the expected ones. This is not just limited by technical advancements but also some intrinsic properties of the device. For instance, in a photon-counting device, one of the significant issues is the dark count, that is, the firing of the detector based on the thermal noise of the environment \cite{darkcount1, darkcount2}. Furthermore, the detectors might not register a photon, referred to as no-click events. Thus, one can only trust their measurement device up to some degree of error. Consequently, it might well happen that Alice's imprecise measurement allows her to observe quantum steerability even if the quantum state is unsteerable. This would be catastrophic for applications of quantum steering for instance in one-sided quantum key distribution \cite{BCW+12, Xin:20, 1sdikey} or one-sided device-independent certification \cite{Supic, Alex, SBK20, Chen, sarkar6, sarkar7, sarkar8}.

In this work, we investigate the effect of not completely trusting Alice in the quantum steering scenario. We consider a general steering functional and compute the variation of its local bound with respect to the trust in Alice. For a note, a similar problem was considered in \cite{Rosset2012, Arminc} for entanglement witnesses where all the measurement devices are trusted. We observe that the effect of trust on the local bound becomes more significant for higher dimensional systems. For our purpose, we first compute the effect of trusting the measurement devices for standard quantum tomography. Not to our surprise, we observe that identifying higher dimensional quantum states using tomography is much more difficult. However, we observe that the error in identifying the state varies almost linearly with respect to trust one puts in their device. Thus, reducing the error in measurement devices considerably improves quantum state tomography. We then consider a specific example of a family of steering inequalities for arbitrary dimensional systems and show that to observe quantum steerability even for ten-dimensional systems would require trusting Alice close to ninety-nine per cent. 

\section{Preliminaries}
Before proceeding towards the main result of this work, let us introduce the notion of quantum steering.

{\it{Quantum steering---}} In this work, we consider the simplest scenario exhibiting quantum steerability, which consists of two parties namely, Alice and Bob who are in two different separated labs and perform some operations on an unknown system which they receive from a preparation device. Bob can choose among $n$ measurements denoted by $B_y$ where $y\in\{1,2,\ldots,n\}$ each of which results in $d$ outcomes labeled by $b\in{0,1,\ldots,d-1}$. The measurement performed by Bob steers the received system with Alice. The steered states are positive operators denoted by $\sigma_{b|y}\in \mathcal{H}_A$ and the collection of these steered states $\sigma=\{\sigma_{b|y}\ s.t.\ b=0,1,\ldots,d-1, y=1,2,\ldots,n\}$ are called assemblages. If the assemblage has a quantum realization then every steered state can be written as, $\sigma_{b|y}=\Tr_B((\I\otimes N_{b|y})\rho_{AB})$ where $\rho_{AB}\in\mathcal{H}_A\otimes\mathcal{H}_B$ denotes the shared state generated by the preparation device. Bob's measurements have a quantum realization that guarantees that $N_{b|y}$ are positive operators and $\sum_{b}N_{b|y}=\I$ [see Fig. \ref{fig1}].

\begin{figure}[t]
    \centering
    \includegraphics[scale=.4]{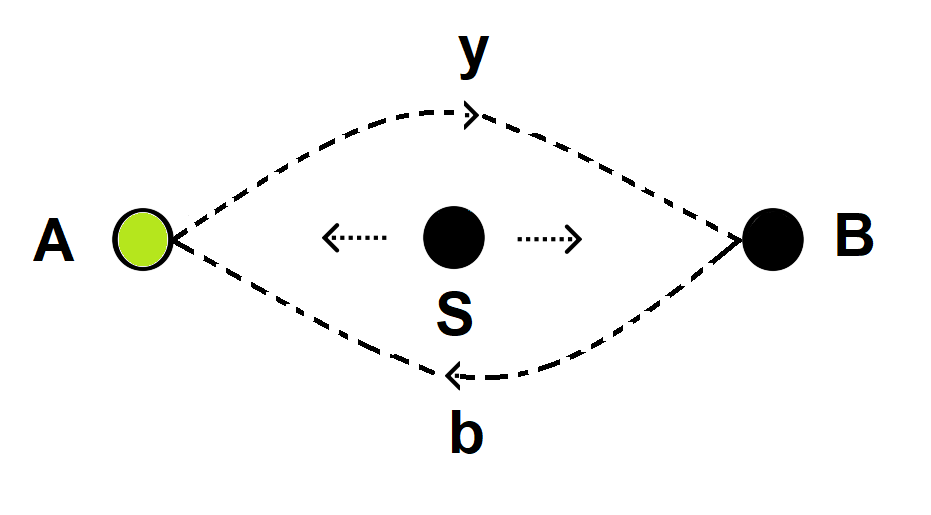}
    \caption{Quantum steering scenario. Alice and Bob are spatially separated and receive subsystems from the source $S$. Alice is trusted here. Alice then sends input $y$ to Bob based on which he performs a measurement on his subsystem and obtains an outcome $b$ which is then sent back to Alice. Since Alice is trusted she then performs a tomography on her subsystem to recover the local assemblage.}
    \label{fig1}
\end{figure}

Alice has complete knowledge about her lab which guarantees that there is no external intruder and her measurement statistics are generated only by the received system. This allows her to simply do a tomography of the received subsystem. If the shared state is not steerable, then the assemblage has a local hidden state model(LHS) defined as,
\begin{eqnarray}\label{LHS1}
\sigma_{b|y}=\sum_\lambda q_\lambda p_\lambda(b|y)\sigma_\lambda
\end{eqnarray}
where $\sum_{\lambda}q_\lambda=1$, $p_\lambda(b|y)$ are probability distribution over $\lambda$ and $\sigma_\lambda$ are density matrices over $\mathcal{H}_B$.

As Alice can perform topographically complete measurements on $\sigma_{b|y}$, in general, quantum steering is witnessed by so-called `steering functional', a map from the assemblage $\{\sigma_{b|y}\}$ to a real number which is defined by a set of matrices $F_{b|y}$ that maps the assemblage to a real number as
\begin{eqnarray}\label{stefn}
    W=\sum_{b,y}\Tr(F_{b|y}\sigma_{b|y}).
\end{eqnarray}
The maximum value $\beta_{L}$ of $W$ that is achievable if the assemblage $\{\sigma_{b|y}\}$ can be described using a LHS model  \eqref{LHS1} is referred to as local bound. Let us now consider the scenario if Alice is not completely trusted.

\section{Results}

{\it{Distrustful quantum steering---}} Consider again the scenario depicted in Fig. \ref{fig1} without completely trusting Alice. From a practical perspective, this represents the case when even trusted devices can become noisy based on different uncontrollable factors. Then, the 'imprecise tomography' of the trusted party can not reveal the ideal assemblage $\{\sigma_{b|y}\}$ generated by Bob's measurement. We show later that this would change the local bound depending on the trust in Alice's devices.

Let us now quantify the amount of trust Alice can make about her devices. For this purpose, consider that the local quantum state belongs to $\mathbb{C}^d$ and a set of tomographically-complete observables $\{\sigma_i\}$, where $i=1,\ldots,d^2$ such that $\sigma_1=\I$, that Alice expects her measurement device to perform. Consider again a set of observables $\{\tau_i\}$ that are actually being performed by the measurement device. Without loss of generality, we assume all the above-described observables are unitary. Then the amount of trust of Alice on her measurement device $\Gamma_A$ is quantified as
\begin{eqnarray}
    \Gamma_A=1-\frac{1}{2d k}\max\sum_{i}||\sigma_i-\tau_i||
\end{eqnarray}
where $||A||=\Tr(AA^{\dagger})$ and $0\leq  \Gamma_A\leq 1$ and $k$ is the number of measurements with Alice. 
It is simple to observe from the above formula that if Alice completely trust her device then $\sigma_i=\tau_i$ for all $i$ which gives us the value $\Gamma_A=1$. If Alice completely untrusts her device and knows nothing about her device then this value can reach a value $\Gamma_A=0$. For instance, when $d=2$ instead of performing $\{Z,X,Y\}$ the measurement device performs $\{X,Y,Z\}$. Furthermore, from one of the state-of-the-art experiments for higher dimensional systems \cite{WPD+18} one can naively deduce that the average fidelity $(f_{\mathrm{avg}})$ of the implemented measurements to the ideal ones are $0.96,0.87,0.81$ for $d=2,4,8$ respectively. Consequently, using the formula
\begin{eqnarray}
    \frac{1}{k}\sum_i||\sigma_i-\tau_i||= 2d(1-f_{\mathrm{avg}}) 
\end{eqnarray}
we get the trustness parameter is $\Gamma_A=f_{\mathrm{avg}}$. Thus, it is evident that as dimension of the system grows the trustness in the device reduces significantly.

Consider now that Alice can almost trust her devices, that is, $||\sigma_i-\tau_i||\leq\varepsilon$ for all $i$ and $\Gamma_A=1-\varepsilon/2d$. Then, we establish the following fact.

\begin{fakt}\label{fact1}
Suppose that a source prepares a state $\tilde{\rho}_A\in\mathbb{C}^d$ and sends it to Alice. If Alice can not completely trust her measurement device, that is,  $||\sigma_i-\tau_i||\leq\varepsilon$ for all $i$, then 
\begin{eqnarray}
    ||\rho_A^{\mathrm{inf}}-\tilde{\rho}_A||\leq d^3(d^2-1)\left(\frac{\varepsilon}{2}+\sqrt{2d\varepsilon}\right)^2.
\end{eqnarray}
where $\rho_A^{\mathrm{inf}}$ is the inferred state from Alice's imprecise tomography.
    
\end{fakt}
\begin{proof}
    Since, $\{\sigma_i\}$ forms a complete basis for matrices acting on $\mathbb{C}^d$, we can express any $\rho_A$ as $\rho_A=\sum_i\mathrm{r}_i\sigma_i$ and $\rho_A^{\mathrm{inf}}=\sum_i\mathrm{q}_i\sigma_i$. As $\rho_A$ is hermitian, $\mathrm{r}_i,\mathrm{q}_i$ are real and less than or equal to $1$. Computing the distance between the states $\rho_A,\rho_A^{\mathrm{inf}}$ gives us
    \begin{eqnarray}\label{rel11}
        ||\rho_A^{\mathrm{inf}}-\tilde{\rho}_A||=||\sum_i(\mathrm{q}_i-\mathrm{r}_i)\sigma_i||=d\sum_i(\mathrm{q}_i-\mathrm{r}_i)^2.
    \end{eqnarray}
    Thus, we now focus on finding an upper bound on the quantity $|\mathrm{q}_i-\mathrm{r}_i|$ for all $i$. As Alice now performs the measurements $\{\tau_i\}$ instead of $\{\sigma_i\}$, we obtain that
    \begin{eqnarray}
        \mathrm{q}_i=\frac{1}{d}\Tr(\rho_A\tau_i^{\dagger})
    \end{eqnarray}
    which on expanding $\rho_A$ gives us
    \begin{eqnarray}
    \mathrm{q}_i=\frac{1}{d}\sum_j\mathrm{r}_j\Tr(\sigma_j\tau_i^{\dagger})
    \end{eqnarray}
    which on rearranging and then taking the real part gives us
    \begin{eqnarray}
   \mathrm{r}_i-\mathrm{q}_i=\mathrm{r}_i\left(1-\frac{1}{d}\mathrm{Re}\Tr(\sigma_i\tau_i^{\dagger})\right)-\frac{1}{d}\sum_{j\ne i}\mathrm{r}_j\mathrm{Re}\Tr(\sigma_j\tau_i^{\dagger}).
    \end{eqnarray}
    Now, taking the absolute value on both sides and using the triangle inequality gives us
    \begin{equation}\label{rel0}
        |\mathrm{r}_i-\mathrm{q}_i|\leq|\mathrm{r}_i|\left|1-\frac{1}{d}\mathrm{Re}\Tr(\sigma_i\tau_i^{\dagger})\right|+\frac{1}{d}\sum_{j\ne i}|\mathrm{r}_j\mathrm{Re}\Tr(\sigma_j\tau_i^{\dagger})|.
    \end{equation}
     Now recalling that $||\sigma_i-\tau_i||\leq\varepsilon$ and expanding the left-hand side gives us
     \begin{eqnarray}
        2d-2\mathrm{Re}\ \Tr(\sigma_i\tau_i^{\dagger})\leq\varepsilon
     \end{eqnarray}
     which can be rearranged to obtain
     \begin{eqnarray}\label{rel1}
         \mathrm{Re}\ \Tr(\sigma_i\tau_i^{\dagger})\geq d-\frac{\varepsilon}{2}.
     \end{eqnarray}
     As $||\sigma_i-\tau_i||\geq0$, we also obtain that 
     \begin{eqnarray}
         \mathrm{Re}\ \Tr(\sigma_i\tau_i^{\dagger})\leq d.
     \end{eqnarray}
     Thus, we can safely conclude that
     \begin{eqnarray}
         0\leq 1-\frac{1}{d}\mathrm{Re}\Tr(\sigma_i\tau_i^{\dagger})\leq\frac{\varepsilon}{2d}.
     \end{eqnarray}
     Furthermore, we also find an upper bound to $\Tr(\sigma_j\tau_i^{\dagger})$ $(j\ne i)$
     for which we expand the term 
     $ ||\sigma_j-\tau_i||$ and use triangle inequality to obtain
     \begin{eqnarray}
       ||\sigma_j-\sigma_i||^{\frac{1}{2}}-||\sigma_i-\tau_i||^{\frac{1}{2}}\leq  ||\sigma_j-\sigma_i+\sigma_i-\tau_i||^{\frac{1}{2}}\nonumber\\ \leq ||\sigma_j-\sigma_i||^{\frac{1}{2}}+||\sigma_i-\tau_i||^{\frac{1}{2}}.
     \end{eqnarray}
     Using the fact that $\Tr(\sigma_i\sigma_j^{\dagger})=0$, we obtain $||\sigma_j-\sigma_i||=2d$ and consequently we have that
     \begin{eqnarray}
          \sqrt{2d}-\sqrt{\varepsilon}\leq||\sigma_j-\tau_i||^{\frac{1}{2}}\leq  \sqrt{2d}+\sqrt{\varepsilon}.
     \end{eqnarray}
     Expanding the above formula gives us
     \begin{eqnarray}
         \left| \mathrm{Re}\ \Tr(\tau_i\sigma_j^{\dagger})+\frac{\varepsilon}{2}\right|\leq \sqrt{2d\varepsilon}
     \end{eqnarray}
     which on further simplification gives us
     
     \begin{eqnarray}\label{rel2}
         | \mathrm{Re}\ \Tr(\tau_i\sigma_j^{\dagger})|\leq \frac{\varepsilon}{2}+\sqrt{2d\varepsilon}\qquad j\ne i .
     \end{eqnarray}
     Finally, considering Eq. \eqref{rel0} and then using Eqs. \eqref{rel1} and \eqref{rel2} gives us
     \begin{eqnarray}
          |\mathrm{r}_i-\mathrm{q}_i|\leq d\left(\frac{\varepsilon}{2}+\sqrt{2d\varepsilon}\right)
     \end{eqnarray}
     where we also used the fact that $|\mathrm{r}_i|\leq1$ for all $i$. Thus, we obtain from Eq. \eqref{rel11} that
     \begin{eqnarray}
         ||\rho_A^{\mathrm{inf}}-\tilde{\rho}_A||\leq d^3(d^2-1)\left(\frac{\varepsilon}{2}+\sqrt{2d\varepsilon}\right)^2.
     \end{eqnarray}
     This completes the proof.
\end{proof}
It is not surprising to observe from the above result that as the dimension of the state grows, the distance between the actual state and the inferred state grows. However, the interesting fact is that the growth scales up extremely fast with respect to the dimension. It is important to note here that we use triangle inequalities to arrive at the above bound and it might be possible to tighten the bounds using numerical methods. It might also be possible to analytically find better bounds using some different inequalities or better quantifiers of the error in the measurements. For a note, the effect of imprecise tomography for quantum states that are locally qubits was considered in Ref. \cite{Rosset2012}.  

Let us now compute the correction to the local bound if one can not completely trust Alice. For this purpose, we consider the general steering functional $W$ \eqref{stefn} and recall that the assemblage $\{\sigma_{b|y}\}$ is reconstructed by Alice using local tomography. We now denote $\{\tilde{\sigma}_{b|y}\}$ as the actual assemblage and $\{\sigma^{\mathrm{inf}}_{b|y}\}$ as the assemblage inferred by Alice. Let us establish the following fact.

\begin{figure*}[t]
    \centering
    \includegraphics[scale=.5]{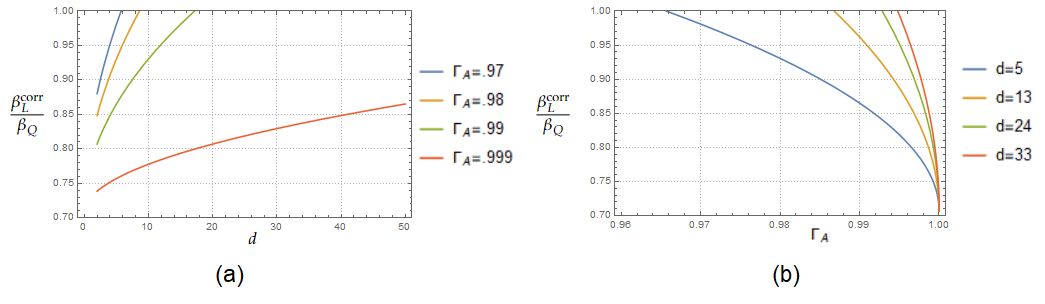}
    \caption{(a)Variation of the upper bound to the ratio of corrected local bound and quantum bound for the steering inequality $\mathcal{B}_d$ with respect to the trust parameter $\Gamma_A$. One can observe that to observe steerability for $d\geq 20$ one has trust Alice for at least $99$ per cent. (b)Variation of the upper bound to the ratio of corrected local bound and quantum bound for the steering inequality $\mathcal{B}_d$ \eqref{SteeIn} with respect to the local dimension $d$ of Alice. One can clearly observe that even trusting Alice $98$ per cent is not enough to witness steerability when $d=10$.   }
    \label{fig2}
\end{figure*}

\begin{fakt}\label{fact2}
   Consider the quantum steering scenario Fig. \ref{fig1} such that Bob sends an assemblage $\{\tilde{\sigma}_{b|y}\}$ to Alice. If Alice can not completely trust her measurement device, that is,  $||\sigma_i-\tau_i||\leq\varepsilon$ for all $i$, then the local bound $\beta_L^{\mathrm{corr}}$ of $W$ due to Alice's imprecise tomography is upper bounded as
\begin{eqnarray}
    \beta_L^{\mathrm{corr}}\leq\beta_L+\mathcal{N}d^2\sqrt{d}\left(\frac{\varepsilon}{2}+\sqrt{2d\varepsilon}\right)
\end{eqnarray}
where $\beta_L$ is the local bound of $W$ if Alice can completely trust her device. Here $\mathcal{N}=\sum_{b,y}p(b|y)\sqrt{\Tr(F_{b|y}F_{b|y}^{\dagger})}$ and $p(b|y)$ is the probability of obtaining outcome $b$ given input $y$ of Bob.
\end{fakt}
\begin{proof}
The idea of the proof is as follows. Consider that Bob prepares an assemblage that attains the local bound of the steering functional $W$. Now, due to the imprecise tomography of Alice, the value of the steering functional might become higher than the local bound $\beta_L$. This would lead Alice to conclude that the assemblage is steerable even if it is not. Thus, we find the upper bound of the steering functional if Alice can not exactly recover the assemblage prepared by Bob. To guarantee the steerability of Bob to Alice, one needs to observe the value higher than the corrected local value $\beta_L^{\mathrm{corr}}$. 

Considering the steering functional $W$ \eqref{stefn} such that the assemblage prepared by Bob $\{\tilde{\sigma}_{b|y}\}$ attains the local bound of $W$ thus giving us
\begin{eqnarray}\label{rel121}
\sum_{b,y}\Tr(F_{b|y}\tilde{\sigma}_{b|y})=\beta_L.
\end{eqnarray}
Now, Alice infers from her tomography that the assemblage is $\{\sigma^{\mathrm{inf}}_{b|y}\}$ thus evaluating $W$ to be
\begin{eqnarray}
W=\sum_{b,y}\Tr(F_{b|y}\sigma^{\mathrm{inf}}_{b|y}).
\end{eqnarray}
The above expression can be rewritten using Eq. \eqref{rel121} as
\begin{eqnarray}\label{rel122}  W=\beta_L+\sum_{b,y}\Tr\left(F_{b|y}\left(\sigma^{\mathrm{inf}}_{b|y}-\tilde{\sigma}_{b|y}\right)\right).
\end{eqnarray}
Recalling that any assemblage can be written as $\sigma_{b|y}=p(b|y)\rho_{b|y}$ where $p(b|y)$ is the probability of obtaining outcome $b$ by Bob when he gives an input $y$. Now, the imprecise tomography of Alice will not change $p(b|y)$. Consequently, utilising fact \ref{fact1} and Holder's inequality for matrix norms given by $|\Tr(AB^{\dagger})|\leq||A||_F||B||_F$ where $||A||_F=\sqrt{||A||}$ for any two matrices $A,B$ gives us
\begin{eqnarray}
W\leq\beta_L+\mathcal{N}d^2\sqrt{d}\left(\frac{\varepsilon}{2}+\sqrt{2d\varepsilon}\right)
\end{eqnarray}
where $\mathcal{N}=\sum_{b,y}p(b|y)\sqrt{\Tr(F_{b|y}F_{b|y}^{\dagger})}$. For simplicity, we used the fact that $d^2-1\leq d^2$ in the above expression. This should be considered as the corrected local bound $\beta_L^{\mathrm{corr}}$ as no assemblage that can be described using an LHS model can cross this bound even if Alice is not completely trusted. This completes the proof.
\end{proof}
Due to the imprecise measurement, one can not conclude that obtaining a value higher than $\beta_L$ is sufficient to witness quantum steering. Thus, to be confident that the assemblage prepared by Bob is steerable one needs to obtain a value higher than $\beta_L^{\mathrm{corr}}$.  It is interesting to note that for any general steering inequality the corrected local bound scales almost as square root of the trust parameter $\Gamma_A$. However, as the dimension of the trusted party grows, the correction term grows more substantially suggesting that observing steering in higher dimensions would require extremely precise measurements by the trusted party. Furthermore, the correction term increases linearly with the number of inputs or outputs of Bob's measurement suggesting the correction is minimum when the number of inputs and outputs of Bob are two. Again let us remark that this bound might is not tight and can be improved via numerical approaches.

In \cite{sarkarcode}, we provide a Mathematica program for finding the bound corrected LHS bound in the simplest scenario corresponding to the case when Alice's system is of dimension two and Bob has two inputs and outputs, that is, $b=y=0,1$. We then follow exactly the same steps as in the proofs of Fact \ref{fact1} and Fact \ref{fact2}. We first find the maximum error in bounding the distance between the ideal state and the inferred state. The program \cite{sarkarcode} is extremely simple and one just needs to input the error $\varepsilon$ in tomography and the parameters of the ideal state. To find the corrected steering bound, one first needs to input the matrices $F_{b|y}$ and the assemblage $\sigma_{b|y}$ that attains the maximum ideal LHS bound. Then, given the error $\varepsilon$, the program returns the corrected local bound. The algorithm can straightaway be generalized to arbitrary dimensions. As an example, we consider the simplest steering inequality where $F_{0|0}=\proj{0}, F_{1|0}=\proj{1},F_{0|1}=\proj{+},F_{1|1}=\proj{-}$ with LHS bound $\beta_L=1+1/\sqrt{2}$. The LHS bound is attained by the assemblage  $\sigma_{0|0}=\proj{\overline{0}}, \sigma_{1|0}=0,\sigma_{0|1}=\proj{\overline{0}},\sigma_{1|1}=0$. Here, $\ket{0},\ket{1}$ denotes the computational basis and $\ket{\pm}=1/\sqrt{2}(\ket{0}\pm\ket{1})$ and $\ket{\overline{0}}=\cos{(\pi/8)}\ket{0}+\sin{(\pi/8)}\ket{1}$. Given arbitrary error in the measurements the program \cite{sarkarcode} returns the tight corrected LHS bound.

Let us now consider a specific family of steering inequalities first constructed in \cite{Horodecki} where the trusted party performs two measurements $Z_d,X_d$ that are mutually unbiased. Mutually unbiased bases for higher dimensional systems have been experimentally implemented \cite{exp1}. Unlike the general case, here Alice does not need to do a full tomography of her assemblage but only perform two measurements on it. Furthermore, the operators $F_{a|x}$ in this case are the measurement operators of Alice. As the number of measurements is lower, so would the correction to the local bound. To find a better upper bound, we evaluate it again without referring to the previous general result, however, the ideas used to find the bound are similar. Here we consider the form of the inequality in the observable picture [see \cite{sarkar6}
for a way to express steering inequality of the form \eqref{stefn} in terms of generalised observables] and compute its corrected local value. The steering inequality is given as

    \begin{eqnarray}\label{SteeIn}
\mathcal{B}_d=\sum_{k=1}^{d-1}\langle Z_d^{(k)}\otimes B_1^{(k)}\rangle+\langle X_d^{(k)}\otimes B_2^{(k)}\rangle
\end{eqnarray}
where 
\begin{equation}
    Z_d=\sum_{i=0}^{d-1}\omega^{i}\ket{i}\!\bra{i},\qquad X_d=\sum_{i=0}^{d-1}\ket{i+1}\!\bra{i}.
\end{equation}

Let us now assume that $||A_1^{(k)}-Z_d^k||\leq\varepsilon$ and $||A_2^{(k)}-X_d^k||\leq\varepsilon$ where $A_1,A_2$ are the observables representing the actual measurements of Alice. The trust parameter $\Gamma_A$ for this case is again $\Gamma_A=1-\varepsilon/2d$. As proven 
 in \cite{sarkar6}, for any inequality of the form
\begin{eqnarray}\label{SteeIn1}
\mathcal{B}_d^{\mathrm{corr}}=\sum_{k=1}^{d-1}\langle A_1^{(k)}\otimes B_1^{(k)}\rangle+\langle A_2^{(k)}\otimes B_2^{(k)}\rangle
\end{eqnarray}
the local bound is given by
\begin{eqnarray}
\beta_L^{\mathrm{corr}}\leq\max_{{\rho}}\sum_{k=1}^{d-1} \left|\langle A_1^{k}\rangle_{\rho}\right|+\left|\langle A_2^{k}\rangle_{\rho}\right|.
\end{eqnarray}
As Alice actually measures $A_1,A_2$ the above bound represents the corrected local bound $\beta_L^{\mathrm{corr}}$. Using triangle inequality we obtain
\begin{eqnarray}
\beta_L^{\mathrm{corr}}\leq\max_{{\rho}}\sum_{k=1}^{d-1} \left|\langle A_1^{k}-Z_d^k\rangle_{\rho}\right|+\left|\langle A_2^{k}-X_d^k\rangle_{\rho}\right|\nonumber\\+\max_{{\rho}}\sum_{k=1}^{d-1} \left|\langle Z_d^k\rangle_{\rho}\right|+\left|\langle X_d^k\rangle_{\rho}\right|.
\end{eqnarray}
The term appearing in the second line of the above formula is in fact the local bound of $\mathcal{B}_d$ and can be inferred from \cite{Horodecki} to be $\sqrt{2}(d-1)$. Evaluating the correction term again using Holder's inequality gives us
\begin{eqnarray}
    \beta_L^{\mathrm{corr}}\leq \max_{{\rho}}\sum_{k=1}^{d-1}(||A_1^{(k)}-Z_d^k||_F+||A_2^{(k)}-X_d^k||_F)||\rho||_F\nonumber\\
    +\sqrt{2}(d-1).\qquad
\end{eqnarray}
Using the fact that $||\rho||_F\leq1$ finally gives us
\begin{eqnarray}\label{bound2}
    \beta_L^{\mathrm{corr}}\leq (d-1)(\sqrt{2}+\sqrt{\varepsilon}).
\end{eqnarray}

To compare the variation of the corrected local bound for different dimensions, we normalise it using the quantum bound of the steering functional \eqref{SteeIn} which is $\beta_Q=2(d-1)$ \cite{sarkar6, Horodecki}. 
In Fig \ref{fig2} we plot the ratio of classical to quantum bound and show that even  $99$ per cent trust with Alice is not convincing to witness any steerability for dimensions greater than $20$, that is the corrected local value becomes equal to the quantum bound. We also find that even for $d=5$ one needs to trust Alice more than $95$ per cent. 
\section{Conclusions}

In this work, we investigate the significance of trust with one of the parties to observe quantum steering. We found analytical upper bounds to the corrected local bound using simple mathematical techniques and found that for any general steering inequality, the corrected LHS bound varies almost square root with respect to the trust but scales up extremely fast with dimension. Thus, observing higher dimensional quantum steering might be much more difficult compared to lower dimensions. In order to find this bound, we also computed the effect of trust in standard quantum tomography. For a remark, it is possible to detect quantum steering where the trust in one party’s measurements is replaced by the trust in that party’s dimension \cite{dimboundsteer1,dimboundsteer2}. This makes it more robust against the errors in the tomography of the assemblages.

Several follow-up problems arise from this work. First, it will be interesting to find other techniques to find the upper bound to the local bound with partially trusted Alice which might be tighter. We believe that for different classes of steering inequalities, one might also find numerical techniques that give much tighter bounds. Secondly, the above techniques might be generalizable to multipartite quantum steering inequalities. Since the problem that is introduced in this work is concerned with the trust one puts in their devices, one can consider similar problems in any semi-device independent scheme, for instance, in prepare and measure scenarios or source independent scenarios. One could also consider the scenario to detect nonlocality using trusted quantum inputs \cite{MDI-STEER1, MDI-STEER} and follow the same steps as described above to relax the trustness on the inputs.  Furthermore, it will be interesting to observe the effects of trust in one-sided device-independent key distribution schemes. Finally, a difficult problem in this regard will be to find steering inequalities for higher dimension systems where the amount of trust is within the current experimental limits.

{\it{Note added.---}} See also a related work \cite{Arminc} that investigates the effect of imprecise measurements by the trusted party on steering inequalities.

\begin{acknowledgments}
 We thank Armin Tavakoli for fruitful discussions. This project was funded within the QuantERA II Programme (VERIqTAS project) that has received funding from the European Union’s Horizon 2020 research and innovation programme under Grant Agreement No 101017733.
\end{acknowledgments}

\providecommand{\noopsort}[1]{}\providecommand{\singleletter}[1]{#1}%

\end{document}